\documentclass[11pt,a4paper]{article}
\usepackage[utf8]{inputenc}
\usepackage[T1]{fontenc}
\usepackage{amsmath,amssymb,amsthm}
\usepackage{mathtools}
\usepackage[margin=1in]{geometry}
\usepackage{graphicx}
\usepackage[british]{babel}
\usepackage{amsfonts}
\usepackage{microtype} 
\usepackage{caption}
\usepackage{appendix}
\usepackage[colorlinks=true, linkcolor=blue, citecolor=blue, urlcolor=blue]{hyperref}

\newtheorem{theorem}{Theorem}[section]
\newtheorem{definition}[theorem]{Definition}
\newtheorem{proposition}[theorem]{Proposition}
\newtheorem{lemma}[theorem]{Lemma}

\theoremstyle{remark}
\newtheorem{remark}[theorem]{Remark}

\newcommand{\R}{\mathbb{R}}
\newcommand{\E}{\mathbb{E}}
\newcommand{\dd}{\mathop{}\!\mathrm{d}}
\newcommand{\Malpha}{M_\alpha}
\newcommand{\gsalpha}{g_s^{(\alpha)}}
\newcommand{\gvalpha}{g_v^{(\alpha)}}

\newcommand{\Zsalpha}{Z_s^{(\alpha)}}
\newcommand{\gsalphaOne}{g_s^{(1)}} 
\newcommand{\gvalphaOne}{g_v^{(1)}} 
\newcommand{\gsalphaTwo}{g_s^{(2)}} 
\newcommand{\hcal}{\mathcal{H}_\alpha}
\newcommand{\pF}{p^F} 
\newcommand{\rhoM}{\rho^{\mathrm{MMSE}}} 

\title{Mixed Fractional Information: Consistency of Dissipation Measures for Stable Laws}
\author{William Cook \\ \textit{University of Bristol}}
\date{16 April 2025} 
\hypersetup{
    pdftitle={Mixed Fractional Information: Consistency of Dissipation Measures for Stable Laws},
    pdfauthor={William Cook},
    pdfsubject={Information Theory, Stable Distributions},
    pdfkeywords={Stable Distributions, Information Theory, Relative Entropy, Fractional de Bruijn Identity, MMSE Estimation, Score Function, Entropy Dissipation, Fisher Score, Consistency, MMSE Score Function, Weighted Functional Inequalities, Information Geometry, Stable Entropy Power Inequality} 
}

\begin{document}
\maketitle
\begin{abstract}
Symmetric $\alpha$‑stable (S\(\alpha\)S) laws with $\alpha<2$ model heavy‑tailed data but render classical Fisher information infinite, hindering standard entropy‑information analyses. Building on Johnson’s fractional de Bruijn framework, we consider the initial rate of relative‑entropy dissipation along the interpolation path \( X_t = (1-t)^{1/\alpha} X_0 + t^{1/\alpha} Z_s^{(\alpha)} \), where \( Z_s^{(\alpha)} \) is an S\(\alpha\)S random variable with scale \(s\). This paper defines and analyses \emph{Mixed Fractional Information} (MFI) specifically for the fundamental case where the initial distribution \(X_0\) is also S\(\alpha\)S, having scale parameter \(v\).

We show that, \emph{in this S\(\alpha\)S-to-S\(\alpha\)S setting}, MFI admits two equivalent calculation methods:
\begin{enumerate}
    \item A chain‑rule expression involving the derivative \(D'(v)\) of the relative entropy \(D(g_v^{(\alpha)} \| g_s^{(\alpha)})\) with respect to the scale parameter \(v\).
    \item An integral representation \(\mathbb{E}_{g_v}\bigl[u(x,0)\,\bigl(\pF_v(x)-\pF_s(x)\bigr)\bigr]\), where \(u(x,0) = x(v-s)/(sv)\) is the specific MMSE-related score function for this linear scale interpolation path, and \(\pF\) denotes the Fisher score.
\end{enumerate}
Our main result is a rigorous proof of the consistency identity
\[
  D'(v)
  = \frac{1}{\alpha\,v}\,
    \mathbb{E}_{g_v}\bigl[X\,\bigl(\pF_v(X)-\pF_s(X)\bigr)\bigr].
\]
This identity mathematically confirms the equivalence of the two MFI calculation methods for S\(\alpha\)S inputs, establishing the internal coherence of the measure within this family. We further prove MFI's non‑negativity (vanishing if and only if \(v=s\)), derive its closed‑form expression for the Cauchy case (\(\alpha=1\)), and validate the consistency identity numerically to \(10^{-6}\) accuracy for general \(\alpha\).

Within the symmetric \(\alpha\)-stable family, MFI thus provides a finite, coherent, and calculable analogue of Fisher information, rigorously connecting entropy dissipation rates to score functions and MMSE-related quantities. This work lays a foundation for exploring potential fractional I-MMSE relations and new (weighted) functional inequalities tailored to heavy-tailed settings.
\end{abstract}


\tableofcontents

\section{Introduction}

\subsection{Motivation: Extending Information Theory to Stable Laws}
\label{sec:motivation} 

Classical information theory finds elegant expression in Gaussian settings. There, Fisher information \(J(h)\) and the de Bruijn identity precisely quantify entropy dynamics under additive Gaussian noise \cite{stam, barron}, underpinning numerous applications in science and engineering. This framework's reach, however, is curtailed when confronting the heavy-tailed phenomena prevalent in diverse fields such as finance, signal processing, and physics. Such phenomena are often modelled effectively by symmetric \(\alpha\)-stable (S\(\alpha\)S) distributions with stability index \(\alpha \in (0, 2)\). For these distributions, the situation changes fundamentally. Their characteristic power-law density decay, \(\gsalpha(x) \sim c_\alpha |x|^{-1-\alpha}\) as \(|x| \to \infty\) (where \(c_\alpha > 0\) is a constant depending on \(\alpha\)), implies infinite variance (for \(\alpha < 2\)) and ensures the integrand defining classical Fisher information, \(h(x) (\pF_h(x))^2\), decays too slowly for convergence over \(\mathbb{R}\). Consequently, \(J(h) = \int_{\mathbb{R}} h(x) (\pF_h(x))^2 \, \dd x\) diverges whenever \(h\) is an S\(\alpha\)S density with \(\alpha<2\), rendering this cornerstone measure unusable and highlighting a critical gap in the standard toolkit.

This divergence is not merely technical; it signifies a fundamental limitation of standard information-theoretic tools in non-Gaussian, heavy-tailed regimes, necessitating new conceptual frameworks. A promising direction stems from the fractional de Bruijn identity framework developed by Johnson \cite{Johnson2018} specifically for S\(\alpha\)S target densities (crucially, covering the full range \(\alpha \in (0, 2]\), including the Gaussian endpoint). Johnson's work provides structures to analyse entropy dissipation beyond the Gaussian paradigm, using dynamics inherent to stable processes.

Building on this foundation, this paper defines and rigorously analyses Mixed Fractional Information (MFI) \emph{specifically for the family of symmetric \(\alpha\)-stable distributions}. Within this context, MFI is derived from the initial rate of relative entropy dissipation along canonical S\(\alpha\)S interpolation paths of the form \(X_t = (1-t)^{1/\alpha} X_0 + t^{1/\alpha} Z_s^{(\alpha)}\). We focus particularly on the fundamental case where both the initial distribution \(X_0\) and the target noise \(Z_s^{(\alpha)}\) correspond to S\(\alpha\)S densities (e.g., with scales \(v\) and \(s\), respectively), investigating MFI's properties and internal consistency \emph{within this setting}. Our analysis establishes MFI as a finite, consistent, and theoretically grounded measure \emph{for comparing pairs of S\(\alpha\)S distributions}, providing a viable alternative exactly where classical Fisher information diverges for this important class of laws. MFI's relationship to other generalisations and alternative approaches is detailed in Section~\ref{sec:related_work}.

\subsection{\texorpdfstring{Symmetric \(\alpha\)-Stable Distributions}{Symmetric alpha-Stable Distributions}}
An S\(\alpha\)S random variable \(\Zsalpha\) is characterised by its characteristic function:
\begin{equation} \label{eq:cf}
\phi_s(k) = \E[e^{ik \Zsalpha}] = e^{-s|k|^\alpha}, \quad k \in \R,
\end{equation}
where \(\alpha \in (0, 2]\) is the stability index and \(s > 0\) is the scale parameter. We denote its probability density function by \(\gsalpha(x)\). We denote the derivative with respect to \(x\) as \(g_s^{(\alpha)\prime}(x)\). Key properties include \cite{samorodnitsky, zolotarev3}:
\begin{itemize}
    \item \textbf{Stability}: If \(Z_u^{(\alpha)} \sim g_u^{(\alpha)}\) and \(Z_v^{(\alpha)} \sim g_v^{(\alpha)}\) are independent, then their sum \(Z_u^{(\alpha)} + Z_v^{(\alpha)} \sim g_{u+v}^{(\alpha)}\). This corresponds to the convolution property \(g_u^{(\alpha)} * g_v^{(\alpha)} = g_{u+v}^{(\alpha)}\).
    \item \textbf{Scaling}: For any constant \(a \in \R\), \(a \Zsalpha \sim g_{|a|^\alpha s}^{(\alpha)}\). Explicitly, $\gvalpha(x) = v^{-1/\alpha} g_1^{(\alpha)}(v^{-1/\alpha} x)$.
\end{itemize}
Special cases relevant to this work are:
\begin{itemize}
    \item \(\alpha = 2\) (Gaussian): \(\gsalphaTwo(x) = (4\pi s)^{-1/2} e^{-x^2 / (4s)}\). This corresponds to a Gaussian distribution with mean 0 and variance \(2s\).
    \item \(\alpha = 1\) (Cauchy): \(\gsalphaOne(x) = \frac{1}{\pi} \frac{s}{s^2 + x^2}\).
\end{itemize}

\paragraph{Symmetry Assumption.}
Throughout this paper, we focus exclusively on \textit{symmetric} \(\alpha\)-stable (S\(\alpha\)S) distributions. Generalisations to asymmetric stable laws are beyond the current scope. See Section \ref{sec:assumptions} for further discussion.

\paragraph{Parameter Range.}
We consider \(\alpha \in (0, 2]\) and scale parameters \(s, v > 0\).

We denote the Fisher score of a density \(f\) by \(\pF_f = f'/f\). For S\(\alpha\)S densities \(\gvalpha\), we may use the shorthand \(\pF_v = \pF_{\gvalpha}\) when the context is clear.

\paragraph{Relative score notation.}
For later brevity, when comparing two S\(\alpha\)S densities
\(g_v^{(\alpha)}\) and \(g_s^{(\alpha)}\), we denote their
\textit{score difference} by
\[
  \delta\pF_{v,s}(x)\;:=\;\pF_{\gvalpha}(x)-\pF_{\gsalpha}(x).
\]

\subsection{Johnson’s Fractional de Bruijn Identity Framework}
Johnson \cite{Johnson2018} introduced an interpolation process between an arbitrary initial random variable \(X_0\) with density \(h\) and an S\(\alpha\)S variable \(\Zsalpha\):
\begin{equation} \label{eq:interp}
X_t = (1-t)^{1/\alpha} X_0 + t^{1/\alpha} \Zsalpha, \quad t \in [0, 1).
\end{equation}
The density \(h_t\) of \(X_t\) evolves according to the following fractional Fokker–Planck-type equation \cite{Johnson2018} of the form:
\begin{equation} \label{eq:pde}
\frac{\partial h_t}{\partial t}(x) = \frac{s}{\alpha (1-t)} \frac{\partial}{\partial x} [h_t(x) u(x, t)],
\end{equation}
Here, \(u(x,t)\) acts as a score function along the interpolation, encoding the direction of mass transport. It is intimately linked to MMSE estimation, as formalised in \cite{Johnson2018} and discussed in Section \ref{sec:discussion_mfi_mmse}. Johnson derived a fractional de Bruijn identity (FdBR) relating the rate of change of relative entropy \(D(h_t \| \gsalpha) = \int h_t \log (h_t / \gsalpha) \, \dd x\) to an integral functional involving \(u(x, t)\) and the Fisher scores \(\pF_f(x) = f'(x)/f(x)\) (where \(f'(x)\) denotes differentiation with respect to \(x\)). Adapting \cite[Theorem 5.1, Eq. 5.1]{Johnson2018}: 
\begin{equation} \label{eq:fdbr}
\frac{d}{dt} D(h_t \| \gsalpha) = -\frac{s}{\alpha (1-t)} \int_{-\infty}^{\infty} h_t(x) u(x, t) (\pF_{h_t}(x) - \pF_{\gsalpha}(x)) \, \dd x.
\end{equation}

\subsection{Related Work and Positioning} \label{sec:related_work}
The challenge of defining finite and meaningful information measures for stable laws, where classical Fisher information diverges, has spurred several research directions. This work focuses on Mixed Fractional Information (MFI), derived directly from the entropy dissipation rate in Johnson's interpolation framework \cite{Johnson2018}. This approach is distinct from, yet complementary to, other notable efforts:

\begin{itemize}
    \item \textbf{Fractional Fisher Information (FFI):} Developed by Toscani and others \cite{Toscani2016}, FFI typically employs fractional derivatives (e.g., Riesz-Feller) in its definition. Relative versions of FFI (\(I_\lambda(X|Z_\lambda)\)) have been proposed to remain finite when comparing a general distribution to a stable law, providing tools for studying entropy inequalities and convergence towards stable limits within a calculus-based framework \cite{Toscani2015b, Toscani2016}. MFI, in contrast, arises from the dynamics of a specific stochastic process (Eq. \ref{eq:interp}) and connects entropy change to an MMSE-related score, rather than directly using fractional derivatives in its definition.

    \item \textbf{Generalised Measures and Estimation Bounds:} Motivated by parameter estimation in impulsive (alpha-stable) noise, researchers like Saad and Fayed \cite{SaadFayed2016} have proposed alternative "alpha-power" concepts intended to replace variance and generalised Fisher information analogues. Their goal is often to establish performance bounds (like Cramer-Rao type inequalities) suitable for stable noise environments, focusing more directly on estimation theory applications.
\end{itemize}
MFI, defined via Eq. \ref{eq:mfi_def}, leverages the specific structure of the stable interpolation path and the associated score function \(u(x, t)\), which incorporates an MMSE-related term \cite{Johnson2018}. This inherent connection to estimation (discussed further in Section \ref{sec:discussion_mfi_mmse}) and its interpretation as an entropy dissipation rate distinguish MFI from FFI and alpha-power frameworks, suggesting it may offer unique insights into the interplay between information dynamics and estimation in stable systems.

\subsection{Contributions} \label{sec:contributions}
Our main contributions within the MFI framework are:
\begin{enumerate}
    \item Defining the Mixed Fractional Information \(\Malpha(h \| \gsalpha)\) via the initial dissipation rate of relative entropy (Section \ref{sec:def_mfi}), leading to two calculation methods: a chain rule form and an integral form.
    \item Deriving the exact analytical form of the initial score function \(u(x, 0) = x (v - s)/(sv)\) for the S\(\alpha\)S-to-S\(\alpha\)S interpolation (Section \ref{sec:score}).
    \item Providing a rigorous analytical proof (Section \ref{sec:consistency}, Appendix \ref{app:consistency_proof}) establishing the mathematical identity \(D'(v) = \frac{1}{\alpha v} \mathbb{E}_{g_v}[X (\pF_v - \pF_s)]\). This identity equates the chain rule and integral formulations of MFI, confirming the framework's internal consistency for S\(\alpha\)S distributions and directly linking entropy derivatives to Fisher score differences.
    \item Proving the non-negativity \(\Malpha(\gvalpha \| \gsalpha) \geq 0\), with equality iff \(v = s\), using fundamental properties of relative entropy (Section \ref{sec:positivity}).
    \item Providing explicit formulas for the Cauchy case (\(\alpha = 1\)) using the established closed-form relative entropy \cite{ChyzakNielsen2019} (Section \ref{sec:cauchy}).
    \item Validating the consistency identity through high-precision numerical simulations using robust methods detailed in Appendix \ref{app:numerical_validation}, confirming the theoretical results (Table \ref{tab:numval_tier2}). This process also clarified that a previously considered Log-Sobolev inequality for the Cauchy case does not hold (Section \ref{sec:discussion_lsi}).
\end{enumerate}
The proven consistency solidifies the MFI calculation framework for S\(\alpha\)S inputs and highlights its structural connection to MMSE-like quantities.

\section{Mixed Fractional Information and Score Function}

\subsection{Definition of MFI and Regularity Conditions} \label{sec:def_mfi}

Classical Fisher information diverges for S\(\alpha\)S distributions with \(\alpha < 2\). Building on Johnson’s FdBR \cite{Johnson2018}, we define MFI via the initial rate of relative entropy dissipation.

\begin{definition}[Regularity Class \texorpdfstring{\(\hcal\)}{H_alpha}] \label{def:hcal}
 The regularity class \(\hcal\) consists of probability densities \(h\) on \(\R\) satisfying sufficient conditions for the MFI framework to be well-posed. Specifically, \(h\) must ensure that:
        \begin{itemize}
            \item \(h(x) > 0\) almost everywhere, \(h \in L^1(\mathbb{R})\), and \(\int |h(x) \log h(x)| \, \dd x < \infty\).
            \item The relative entropy \(D(h \| \gsalpha) = \int h \log(h/\gsalpha) \, \dd x\) is finite.
            \item The density \(h_t\) of the interpolated variable \(X_t\) (Eq. \ref{eq:interp}) exists for \(t\) in some interval \([0, \epsilon)\).
            \item The function \(t \mapsto D(h_t \| \gsalpha)\) is differentiable at \(t=0\).
            \item The Fisher score \(\pF_h = h'/h\) exists (requiring appropriate smoothness of \(h\)).
            \item The score function \(u(x,0)\) associated with the interpolation starting from \(h\) exists and is well-behaved.
            \item The integral defining the FdBR at \(t=0\) (Eq. \ref{eq:fdbr}) converges absolutely: \(\int h(x) |u(x, 0) (\pF_{h}(x) - \pF_{\gsalpha}(x))| \, \dd x < \infty\).
        \end{itemize}
        Symmetric \(\alpha\)-stable densities \(\gvalpha\) satisfy these conditions. Specifically:
        \begin{itemize}
            \item Finiteness of \(D(\gvalpha \| \gsalpha)\): The integrand \(g_v^{(\alpha)}(x)\log(g_v^{(\alpha)}(x)/g_s^{(\alpha)}(x))\) decays like \(O(|x|^{-1-\alpha})\) as \(|x|\to\infty\) (since the log term approaches a constant) and is integrable near \(x=0\). Thus, the KL divergence is finite \cite{samorodnitsky}.
                    \item Differentiability of \(t \mapsto D(h_t \| \gsalpha)\) at \(t=0\): For \(h_t = g_{v(t)}^{(\alpha)}\), differentiating \(D(v(t))\) under the integral is justified by the Dominated Convergence Theorem.  Indeed,
        \[
           \frac{\partial}{\partial v}\Bigl[g_v(x)\log\frac{g_v(x)}{g_s(x)}\Bigr]
           =\frac{\partial g_v}{\partial v}(1+\log(g_v/g_s))
           =O\bigl(|x|^{-1-\alpha}\bigr)
        \]
        uniformly for \(v\) in any compact subset of \((0,\infty)\), and since \(|x|^{-1-\alpha}\) is integrable at infinity for \(\alpha>0\), the interchange is valid (see Appendix \ref{app:consistency_proof}, Step 2).

            \item Smoothness and decay properties of S\(\alpha\)S densities \cite{samorodnitsky, zolotarev3} also ensure the existence of Fisher scores and the convergence of the FdBR integral at \(t=0\).
        \end{itemize}
        Thus, \(\gvalpha \in \hcal\). Further characterisation of \(\hcal\) for general densities is discussed in Section \ref{sec:assumptions}.
\end{definition}

\begin{definition}[Mixed Fractional Information] \label{def:mfi}
For \(h \in \hcal\), \(\gsalpha\) (scale \(s\)), and \(X_t\) as in Eq. \ref{eq:interp} with \(h_0=h\):
\begin{equation} \label{eq:mfi_def}
\Malpha(h \| \gsalpha) = -\frac{\alpha}{s} \frac{d}{dt} D(h_t \| \gsalpha) \bigg|_{t=0}.
\end{equation}
This definition assumes \(h \in \hcal\), the regularity class defined in Definition \ref{def:hcal}, ensuring differentiability of entropy flow and well-posedness of the interpolation process. It isolates the integral term from the FdBR (Eq. \ref{eq:fdbr}) at \(t=0\) and scales it, providing a measure directly related to the initial dissipation mechanism inherent in the stable interpolation process. Intuitively, \(\Malpha(h \| \gsalpha)\) measures how rapidly the distribution \(h_t\) initially diverges (in relative entropy) from \(\gsalpha\) under the fractional-noise interpolation process. It generalises Fisher Information to the heavy-tailed setting.
\end{definition}

\subsection{Connection to Integral Form}
From the FdBR (Eq. \ref{eq:fdbr}) evaluated at \(t=0\), the MFI corresponds to the integral:
\begin{equation} \label{eq:mfi_integral}
\Malpha(h \| \gsalpha) = \int_{-\infty}^{\infty} h(x) u(x, 0) (\pF_{h}(x) - \pF_{\gsalpha}(x)) \, \dd x,
\end{equation}
Note: \(u(x,0)\) depends on the interpolation path originating from the specific initial density \(h\). An explicit form is derived in Lemma \ref{lem:score} only for the case \(h = \gvalpha\). For general \(h\), the score \(u(x,0)\) must be computed from the dynamics of \(h_t\).

\subsection{Analytical Score Function for S\texorpdfstring{$\alpha$}{alpha}S Initial Density} \label{sec:score}

Consider the interpolation \(X_t = (1-t)^{1/\alpha} X_0 + t^{1/\alpha} \Zsalpha\) starting from \(X_0 \sim \gvalpha\).
Using the scaling and stability properties, the density of \(X_t\) is \(h_t = g_{v(t)}^{(\alpha)}\), where the effective scale evolves as:
\begin{equation} \label{eq:vt_def}
v(t) = (1-t)v + ts.
\end{equation}

\begin{figure}[htbp]
  \centering
  \includegraphics[width=0.75\linewidth]{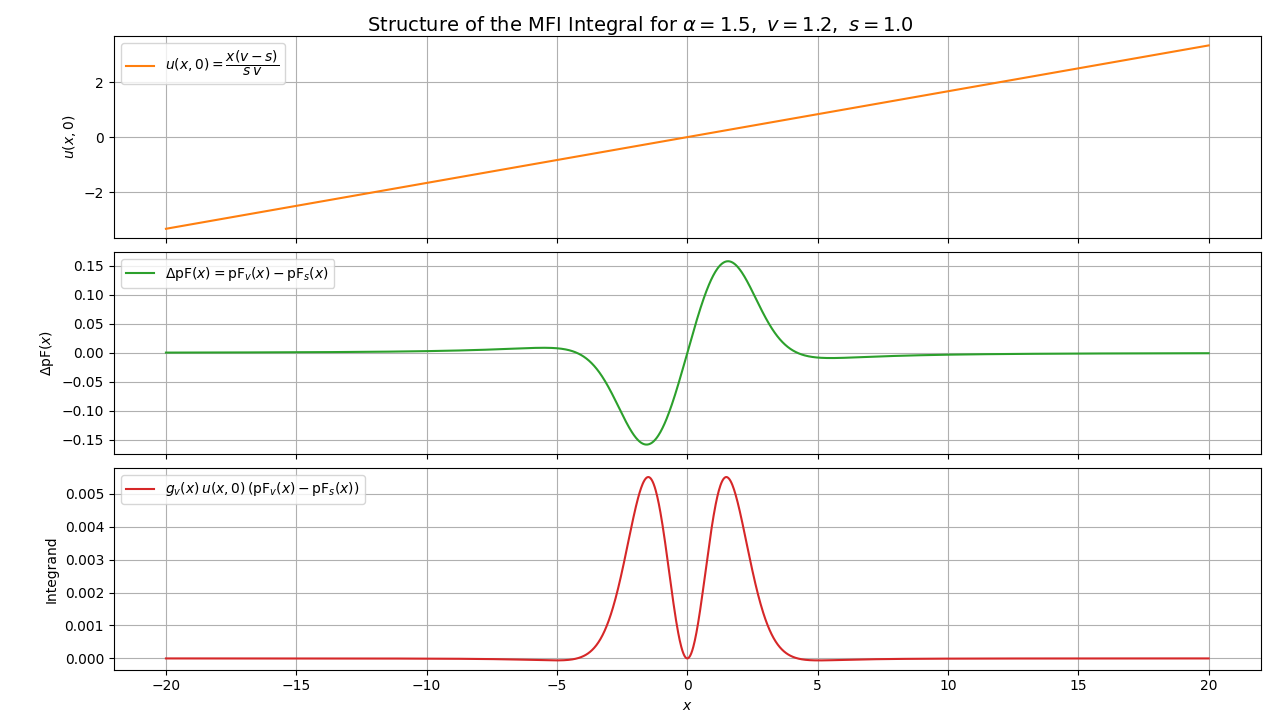}
  \caption{Structure of the MFI integral for symmetric $\alpha$-stable densities under scale interpolation, with $\alpha = 1.5$, initial scale $v = 1.2$, and target scale $s = 1.0$. Top: the MMSE-related score $u(x, 0) = \frac{x (v - s)}{s v}$. Middle: the difference of Fisher scores $\Delta\mathrm{pF}(x) = \mathrm{pF}_{v}(x) - \mathrm{pF}_{s}(x)$. Bottom: the integrand $g_v(x)\, u(x, 0)\, \Delta\mathrm{pF}(x)$ whose integral over $\mathbb{R}$ gives the MFI.}
  \label{fig:structure_mfi_integral}
\end{figure}

\begin{lemma} \label{lem:score}
For the interpolation starting at \(h_0 = \gvalpha\), the score function \(u(x, t)\) from Eq. \ref{eq:pde} at time \(t\) is:
\[ u(x, t) = x \left( \frac{1}{s} - \frac{1}{v(t)} \right), \]
and the initial score function is:
\[ u(x, 0) = \frac{x (v - s)}{sv}. \]
\end{lemma}

\begin{proof}
Following Johnson \cite{Johnson2018}, the score is \(u(x, t) = \rhoM_{X,t}(x) + x/s\), where \(\rhoM_{X,t}(x) = -\E[t^{1/\alpha} \Zsalpha | X_t = x] / (st)\) is related to the Minimum Mean Square Error (MMSE) estimate of the scaled noise component \(t^{1/\alpha}\Zsalpha\) given \(X_t\). Johnson's Lemma 3.3 \cite{Johnson2018} utilises the structure of jointly S\(\alpha\)S variables. Specifically, even if marginal first moments are infinite (i.e., for \(\alpha \le 1\)), the conditional expectation exists in the sense of linear regression, \(\mathbb{E}[t^{1/\alpha}\Zsalpha | X_t = x] = c \cdot x\) for some constant \(c\), due to the linear regression property of jointly S\(\alpha\)S vectors (see \cite[Sec 2.8]{samorodnitsky}). In this specific interpolation context, this linear relationship yields \(\rhoM_{X,t}(x) = -x/v(t)\). Therefore, \(u(x, t) = -x/v(t) + x/s\).     Evaluating at \(t=0\) gives \(u(x, 0) = x(1/s - 1/v) = x(v-s)/(sv)\).
\end{proof}

\begin{remark}[On “conditional expectation” for \(\alpha\le1\)]
Even though for \(\alpha\le1\) the classical \(\mathbb{E}[\,\Zsalpha\mid X_t=x]\) may not exist in \(L^1\), the quantity \(\rhoM_{X,t}(x)\) is understood as the stable‑law linear‑regression (covariation) estimate (see Johnson 2018).  This projection is well‑defined for jointly S\(\alpha\)S vectors and is what enters the FdBR framework.
\end{remark}

\begin{remark}[Specificity of the Linear Score]
The derived linear form of the score function, \(u(x, t) = x(1/s - 1/v(t))\), is specific to the case where the initial distribution \(h_0\) is also S\(\alpha\)S. This result relies on Lemma 3.3 in \cite{Johnson2018}, which leverages the unique properties of jointly S\(\alpha\)S variables within the interpolation structure. For a general initial density \(h \in \hcal\), the score function \(u(x, t)\), and particularly \(u(x, 0)\), would typically be non-linear and depend explicitly on the form of \(h\).
\end{remark}

\section{\texorpdfstring{Properties of MFI for S\(\alpha\)S Densities}{Properties of MFI for S-alpha-S Densities}}

\subsection{Chain Rule Formulation} \label{sec:chain_rule}
When the initial density is S\(\alpha\)S, i.e., \(h_0 = \gvalpha\),
the stability property ensures the interpolated density remains
within the same family: \(h_t = g_{v(t)}^{(\alpha)}\), where the
scale parameter evolves linearly as \(v(t) = v(1-t) + st\).
Let \(D(v')\) denote the relative entropy
\(D(g_{v'}^{(\alpha)} \| \gsalpha)\). As established in the preamble
of Appendix~\ref{app:consistency_proof}, \(D(v')\) is differentiable
with respect to the scale parameter \(v'>0\). Since \(v(t)\) is a
smooth function of \(t\), the chain rule applies to the time
derivative of the relative entropy along the path:
\( \frac{d}{dt}D(h_t \| \gsalpha) = \frac{d}{dt} D(v(t)) \).
Specifically,
\[ \frac{d}{dt} D(v(t)) = D'(v(t)) \frac{d v(t)}{dt} = D'(v(t)) (s - v). \]

\medskip 
\noindent\textbf{Important Limitation.} This chain-rule expression
holds \emph{only} because the interpolation starts from
\(h_0 = g_v^{(\alpha)}\), ensuring that \(h_t=g^{(\alpha)}_{v(t)}\)
remains within the S\(\alpha\)S family throughout the evolution
\(t \in [0, 1)\). If \(h_0\) were an arbitrary density in
\(\mathcal{H}_\alpha\), then \(h_t\) would generally \emph{not}
be an S\(\alpha\)S density, the map
\(t\mapsto D(h_t\|g_s)\) could not be written as a composition
\(D(v(t))\), and the simple chain-rule factorisation
\(D'(v(t))v'(t)\) would no longer be applicable.
\medskip 

Evaluating the valid chain rule expression at \(t=0\) gives
\(\frac{d}{dt} D(h_t \| \gsalpha) \big|_{t=0} = D'(v) (s - v)\).
Substituting this initial rate into the MFI definition
(Eq.~\ref{eq:mfi_def}):
\begin{equation} \label{eq:mfi_exact}
\Malpha(\gvalpha \| \gsalpha) = -\frac{\alpha}{s} [D'(v) (s - v)]
= \frac{\alpha D'(v) (v - s)}{s}.
\end{equation}
This provides the Mixed Fractional Information, for the specific case
of an S\(\alpha\)S initial density compared to an S\(\alpha\)S target,
directly in terms of the derivative of their relative entropy with
respect to the scale parameter.

\subsection{Consistency of Formulations} \label{sec:consistency}
We now establish that the chain rule form (Eq. \ref{eq:mfi_exact}) and the integral form (Eq. \ref{eq:mfi_integral} with the derived score \(u(x, 0)\)) yield the same value. Equating the chain rule form (\(\Malpha^{\text{chain}}\)) with the integral form (\(\Malpha^{\text{integral}}\)) using the derived score \(u(x, 0)\) from Lemma \ref{lem:score} requires establishing the following identity. Substituting \(u(x, 0) = \frac{x (v - s)}{sv}\) from Lemma \ref{lem:score} into Eq. \ref{eq:mfi_integral} with \(h=\gvalpha\):
\[ \Malpha^{\text{integral}} = \int_{-\infty}^{\infty} \gvalpha(x) \cdot \frac{x (v - s)}{sv} \cdot (\pF_{\gvalpha}(x) - \pF_{\gsalpha}(x)) \, \dd x = \frac{v - s}{sv} \int_{-\infty}^{\infty} x \gvalpha(x) (\pF_{\gvalpha}(x) - \pF_{\gsalpha}(x)) \, \dd x. \]
Equating this to the chain rule form, \(\Malpha^{\text{chain}} = \frac{\alpha D'(v) (v - s)}{s}\):
\[ \frac{\alpha D'(v) (v - s)}{s} = \frac{v - s}{sv} \int_{-\infty}^{\infty} x \gvalpha(x) (\pF_{\gvalpha}(x) - \pF_{\gsalpha}(x)) \, \dd x \]
Assuming \(v \neq s\), we can divide by \((v-s)/s\):
\[ \alpha D'(v) = \frac{1}{v} \int_{-\infty}^{\infty} x \gvalpha(x) (\pF_{\gvalpha}(x) - \pF_{\gsalpha}(x)) \, \dd x \]
\[ \implies D'(v) = \frac{1}{\alpha v} \int_{-\infty}^{\infty} x \gvalpha(x) (\pF_{\gvalpha}(x) - \pF_{\gsalpha}(x)) \, \dd x. \]
This identity holds trivially if \(v=s\) since both sides are zero (\(D'(s)=0\) as established in Theorem \ref{thm:pos} proof, and \(\pF_{\gvalpha}=\pF_{\gsalpha}\)).

\begin{proposition}[Consistency Identity] \label{prop:consistency_id}
For S\(\alpha\)S densities \(\gvalpha\) and \(\gsalpha\) with \(\alpha \in (0, 2]\) and \(v, s > 0\), the following identity holds:
\[ D'(v) = \frac{1}{\alpha v} \int_{-\infty}^{\infty} x \gvalpha(x) (\pF_{\gvalpha}(x) - \pF_{\gsalpha}(x)) \, \dd x, \]
where \(D(v) = D(\gvalpha \| \gsalpha)\) and \(\pF_f = f'/f\) (specifically \(\pF_{\gvalpha} = g_v^{(\alpha)\prime} / \gvalpha\) and \(\pF_{\gsalpha} = g_s^{(\alpha)\prime} / \gsalpha\)).
\end{proposition}
\begin{proof}
See Appendix \ref{app:consistency_proof}. For a rigorous numerical validation of this identity, see Appendix \ref{app:numerical_validation}.
\end{proof}
This proposition rigorously confirms the equivalence of the two calculation methods for \(\Malpha(\gvalpha \| \gsalpha)\).

\begin{remark}[Scope of the Consistency Identity]
\label{rem:scope_identity}
The identity established in Proposition~\ref{prop:consistency_id},
and thus the proven equivalence between the chain-rule formulation
of MFI (Eq.~\ref{eq:mfi_exact}) and the integral formulation
(Eq.~\ref{eq:mfi_integral} combined with the score from
Lemma~\ref{lem:score}), crucially \emph{requires} the initial
density to be S\(\alpha\)S, i.e., \(h_0=g^{(\alpha)}_{v}\).
This specific choice guarantees that the interpolation path remains
within the S\(\alpha\)S family (\(h_t=g^{(\alpha)}_{v(t)}\)),
that the initial score function takes the linear form
\(u(x,0)=x(v-s)/(sv)\), and that the time derivative of the
relative entropy follows the chain rule
\(\frac{\dd}{\dd t}D(h_t\|\gsalpha)=D'(v(t))\,v'(t)\).
For a general initial density \(h\in\mathcal{H}_\alpha\), the
interpolated density \(h_t\) typically leaves the S\(\alpha\)S
family, the initial score \(u(x,0)\) is generally non-linear,
and the simple chain-rule factorisation fails. Consequently,
the two MFI calculation methods derived here are not guaranteed
to agree in the general case.
\end{remark}

\subsection{Positivity} \label{sec:positivity}
\begin{theorem} \label{thm:pos}
For any \(\alpha \in (0, 2]\) and scale parameters \(v, s > 0\), the Mixed Fractional Information satisfies
\[ \Malpha(\gvalpha \| \gsalpha) \geq 0, \]
with equality holding if and only if \(v = s\).
\end{theorem}
\begin{proof}
The proof proceeds from the chain rule formulation derived in Section~\ref{sec:chain_rule} (Eq.~\ref{eq:mfi_exact}), which is valid for the S\(\alpha\)S-to-S\(\alpha\)S comparison:
\[ \Malpha(\gvalpha \| \gsalpha) = \frac{\alpha D'(v) (v - s)}{s}. \]
Here, \(D(v') := D(g_{v'}^{(\alpha)} \| \gsalpha)\) is the relative entropy (KL divergence) between two S\(\alpha\)S densities with scales \(v'\) and \(s\), respectively, and \(D'(v)\) is its derivative with respect to \(v'\) evaluated at \(v' = v\).

By the fundamental properties of relative entropy (Gibbs' inequality), \(D(v') \ge 0\) for all \(v' > 0\) \cite{CoverThomas1991}. Furthermore, equality \(D(v')=0\) holds if and only if the two probability densities are identical almost everywhere, i.e., \(g_{v'}^{(\alpha)} = g_s^{(\alpha)}\) a.e.

For the family of symmetric \(\alpha\)-stable distributions considered in this paper (zero location, zero skewness), the mapping from the positive scale parameter (\(v'\) or \(s\)) to the corresponding probability density function (\(g_{v'}^{(\alpha)}\) or \(g_s^{(\alpha)}\)) is injective. This means that distinct positive scale parameters necessarily yield distinct probability density functions. Therefore, the condition \(g_{v'}^{(\alpha)} = g_s^{(\alpha)}\) a.e. is met if and only if the scale parameters are identical: \(v' = s\) \cite[cf. Property 1.2.3]{samorodnitsky}\cite{zolotarev3}.

This establishes that the function \(f(v') = D(v')\) has a unique global minimum at \(v' = s\), where \(D(s) = 0\). As justified in Appendix~\ref{app:consistency_proof} (Preamble based on properties from \cite{samorodnitsky, zolotarev3}), the relative entropy \(D(v')\) is differentiable with respect to \(v'\) for \(v'>0\). Standard calculus then implies that for a differentiable function with a unique global minimum at \(s\), its derivative must be zero at the minimum (\(D'(s) = 0\)), negative to the left (\(D'(v) < 0\) for \(v < s\)), and positive to the right (\(D'(v) > 0\) for \(v > s\)).

Consequently, the product \(D'(v)(v-s)\) is strictly positive whenever \(v \ne s\) (as both factors are negative if \(v < s\), and both are positive if \(v > s\)) and is zero only when \(v = s\). Since the stability index \(\alpha\) and the target scale parameter \(s\) are both strictly positive (\(\alpha \in (0, 2]\), \(s > 0\)), the overall factor \(\alpha/s\) in the MFI expression is positive.

We conclude that \(\Malpha(\gvalpha \| \gsalpha) = (\alpha/s) [D'(v)(v-s)] \geq 0\), with equality holding if and only if \(v=s\).
\end{proof}

\begin{remark}[Note on Convexity]
While the convexity of \(D(v) = D(g_v^{(\alpha)} \| g_s^{(\alpha)})\) as a function of \(v\) is plausible (and holds for \(\alpha=1, 2\)), the positivity proof presented here relies only on the fundamental property that \(D(v)\) has a unique minimum at \(v=s\) and is differentiable. Proving convexity for general \(\alpha \in (0, 2)\) remains an open question but is not necessary for establishing the non-negativity of MFI.
\end{remark}

\subsection{\texorpdfstring{Cauchy Case (\(\alpha = 1\))}{Cauchy Case (alpha = 1)}} \label{sec:cauchy}
For the Cauchy distribution (\(\alpha = 1\)), the relative entropy is known \cite[Theorem 1]{ChyzakNielsen2019}:
\[ D(\gvalphaOne \| \gsalphaOne) = \log \frac{(v + s)^2}{4vs}. \]
The derivative with respect to \(v\) is:
\[ D'(v) = \frac{d}{dv} \left[ 2\log(v+s) - \log(v) - \log(4s) \right] = \frac{2}{v+s} - \frac{1}{v} = \frac{v - s}{v(v + s)}. \]
Substituting \(\alpha=1\) and this \(D'(v)\) into the MFI formula (Eq. \ref{eq:mfi_exact}):
\[ M_1(\gvalphaOne \| \gsalphaOne) = \frac{1 \cdot D'(v) (v - s)}{s} = \frac{1}{s} \left( \frac{v - s}{v(v + s)} \right) (v - s) = \frac{(v - s)^2}{sv(v + s)}. \]
This provides an explicit, closed-form expression for MFI in the Cauchy case, allowing for direct analysis in this setting.
This closed-form expression aligns with the general consistency identity proven in Appendix \ref{app:consistency_proof}, which holds for all \(\alpha \in (0, 2]\), including the Cauchy case (\(\alpha=1\)).

\section{Discussion} \label{sec:discussion}

\subsection{MFI as a Dissipation Measure}
The MFI, \(\Malpha(\gvalpha \| \gsalpha) = \frac{\alpha D'(v) (v - s)}{s}\), measures the initial rate of relative entropy dissipation as \(\gvalpha\) evolves toward \(\gsalpha\) under the specific stable interpolation (Eq. \ref{eq:interp}), scaled appropriately by \(\alpha/s\). Unlike classical Fisher information, which diverges for \(\alpha < 2\), MFI remains finite and provides a well-defined measure reflecting the information dynamics inherent to stable laws within this specific process.

\subsection{Significance of Consistency}
The proven equivalence between the chain rule form (\(\propto D'(v)(v-s)\)) and the integral form (\(\propto \mathbb{E}_{g_v}[u(x,0)(\pF_v - \pF_s)]\)) is crucial. It validates the internal mathematical coherence of the MFI framework when comparing S\(\alpha\)S distributions. It confirms that the derived score function \(u(x,0) = x(v-s)/(sv)\) correctly links the local Fisher score differences to the global change in relative entropy \(D'(v)\) via the identity in Proposition \ref{prop:consistency_id}. This identity itself, which can be written as \(\alpha v D'(v) = \mathbb{E}_{g_v}[X (\pF_v - \pF_s)]\), represents a potentially novel contribution to the analytical understanding of information measures for stable laws, directly connecting entropy derivatives to score functions.

\subsection{Link to MMSE Estimation and I-MMSE Analogy} \label{sec:discussion_mfi_mmse}
The score function \(u(x, t) = \rhoM_{X,t}(x) + x/s\), central to the Fractional de Bruijn Identity (FdBR) and the integral formulation of MFI, explicitly incorporates Johnson’s MMSE-related score \(\rhoM_{X,t}(x)\). This term quantifies aspects of the error in estimating the scaled stable noise component \(t^{1/\alpha}\Zsalpha\) given the interpolated process \(X_t=x\) \cite{Johnson2018}, even when classical moments diverge.

The proven consistency identity (Proposition~\ref{prop:consistency_id}) for the S\(\alpha\)S-to-S\(\alpha\)S interpolation at \(t=0\) strongly reinforces this connection between information dissipation and estimation. The identity can be suggestively rearranged using the score difference notation (\(\delta\pF_{v,s} = \pF_v - \pF_s\)) as:
\[
  \alpha\,v\,D'(v)\;=\;
  \E_{g_v}\!\bigl[X\,\delta\pF_{v,s}(X)\bigr].
\]
This structure echoes the celebrated I-MMSE relation for Gaussian channels \cite{guo}, typically written as
\(
  \frac{\dd}{\dd\mathrm{snr}} I(X; \sqrt{\mathrm{snr}}Z+X)
  = \frac{1}{2} \mathrm{mmse}(\mathrm{snr})
\), where \(Z\) is standard Gaussian noise.
In the Gaussian case, the derivative of mutual information (\(I\), an information measure) with respect to the signal-to-noise ratio (\(\mathrm{snr}\)) equals half the minimum mean-square error (\(\mathrm{mmse}\), an estimation-theoretic quantity).

In our S\(\alpha\)S context, \(D(v) = D(g_v^{(\alpha)} \| g_s^{(\alpha)})\) acts as an \emph{information potential} (relative entropy) dependent on the scale parameter \(v\). Its derivative \(D'(v)\), scaled by \(\alpha v\), is equated to the expectation \(\E_{g_v}[X\,\delta\pF_{v,s}(X)]\). This expectation involves the score difference \(\delta\pF_{v,s}\) and is directly related through the MFI framework (specifically Lemma~\ref{lem:score} and Eq.~\ref{eq:mfi_integral}) to the MMSE-related score \(u(x,0)\), thus serving as an estimation-theoretic proxy reflecting the strength of the dynamics driving \(g_v^{(\alpha)}\) towards \(g_s^{(\alpha)}\).

This structural resonance, arising naturally from the FdBR framework and solidified by the consistency result, suggests that MFI may serve as a valuable tool for investigating estimation-theoretic properties in stable noise environments. It motivates exploring a potential "fractional I-MMSE" framework, further connecting information dynamics and estimation theory in heavy-tailed settings and potentially leading to new performance bounds or insights.

\subsection{Functional Inequalities and Concentration} \label{sec:discussion_lsi}
A key direction for future work is establishing functional inequalities relating the relative entropy \(D(\gvalpha \| \gsalpha)\) to the Mixed Fractional Information \(\Malpha(\gvalpha \| \gsalpha)\). Such inequalities are crucial for understanding concentration properties and convergence rates within the stable framework, analogous to how classical Log-Sobolev (LSI) or Poincaré inequalities function in the Gaussian setting.

An initial investigation might explore a direct LSI analogue, particularly for the Cauchy case (\(\alpha=1\)) where closed forms exist. One could conjecture a relationship like \(D(\gvalphaOne \| \gsalphaOne) \leq C \cdot M_1(\gvalphaOne \| \gsalphaOne)\) for some universal constant \(C\). However, straightforward analysis reveals this simple, unweighted LSI does not hold. Using the explicit expressions derived in Section~\ref{sec:cauchy}, the ratio
\[
 \frac{D(\gvalphaOne \| \gsalphaOne)}{M_1(\gvalphaOne \| \gsalphaOne)} =
 \frac{sv(v+s)}{(v-s)^2} \log \frac{(v+s)^2}{4vs}
\]
is unbounded. For instance, consider the limit as the ratio \(v/s \to \infty\). In this regime, \(D \sim \log(v/s)\) grows logarithmically, while \(M_1 = (v-s)^2 / (sv(v+s)) \sim v^2 / (s v^2) = 1/s\) approaches a constant (assuming fixed \(s\)). Consequently, the ratio \(D/M_1\) diverges as \(v/s \to \infty\) (and similarly as \(v/s \to 0\)).

\medskip
\noindent
This failure aligns with the established fact that symmetric \(\alpha\)-stable laws for \(\alpha \in (0, 2)\) do not possess a spectral gap in \(L^2\) spaces weighted by the density itself, unlike the Gaussian case (\(\alpha=2\)) which corresponds to the Ornstein-Uhlenbeck process (see, e.g.,~\cite{GentilGuillinMiclo2005} and related works on functional inequalities for Lévy processes). The absence of such a gap generally precludes standard (unweighted) Poincaré or log-Sobolev inequalities, thus motivating the investigation of weighted functional inequalities as discussed previously.
\medskip

Indeed, research suggests that appropriate inequalities for heavy-tailed distributions might need to be weighted or otherwise generalised. For instance, weighted Poincaré and LSI inequalities have been successfully established for Cauchy measures, albeit using specific weight functions chosen to counteract the heavy tails \cite{BobkovLedoux2009}. Connections between suitable weighted inequalities and concepts like Stein kernels have also been explored for distributions with heavy tails \cite{Saumard2019}.

While alternative approaches using Fractional Fisher Information have yielded candidate functional inequalities for S\(\alpha\)S laws \cite{Toscani2016}, finding and proving corresponding inequalities (potentially weighted LSI or Poincaré-type) involving MFI remains a significant open problem. Such inequalities would be central to understanding concentration and convergence properties within the MFI framework and would likely need to explicitly account for the \(\alpha\)-dependent structure revealed by the MFI definition and the consistency identity (Proposition~\ref{prop:consistency_id}). The analytical challenge posed by finding appropriate functional inequalities stands in contrast to the successful high-precision verification of the core consistency identity itself, which was robustly confirmed numerically as detailed in Appendix~\ref{app:numerical_validation}.

\subsection{Numerical Confirmation of Consistency} \label{sec:numerical_confirmation}
The analytical proof establishing the consistency between the chain rule and integral formulations of MFI for S\(\alpha\)S inputs (Proposition \ref{prop:consistency_id}) is strongly corroborated by numerical validation. Using a robust computational pipeline detailed in Appendix \ref{app:numerical_validation}---employing adaptive quadrature over infinite domains, stable score computation via log-derivatives, and high-order finite difference methods---we confirmed the identity \(D'(v) = \frac{1}{\alpha v} \mathbb{E}_{g_v}[X (\pF_v - \pF_s)]\) to high precision. Tier 1 validation for Cauchy (\(\alpha=1\)) and Gaussian (\(\alpha=2\) cases achieved machine precision agreement, leveraging known analytic forms. Crucially, Tier 2 validation for the general case (\(\alpha=1.5\)) also yielded excellent agreement (relative error \(< 10^{-6}\)), as shown in Table \ref{tab:numval_tier2}. This demonstrates that the identity holds robustly and is computationally verifiable when appropriate numerical methods are employed, resolving the larger discrepancies observed in initial, less sophisticated numerical tests (which suffered from issues like domain truncation and unstable score estimation). This numerical confirmation solidifies the MFI framework's internal coherence.

\section{Assumptions and Limitations} \label{sec:assumptions}
This work operates under specific assumptions and has certain limitations:
\begin{itemize}
    \item \textbf{Symmetry:} We exclusively consider symmetric \(\alpha\)-stable (S\(\alpha\)S) distributions. Extension to asymmetric stable laws is a potential future direction but introduces additional complexity (e.g., drift terms, modified score functions).
    \item \textbf{Regularity Class \(\hcal\):} The definition of MFI (Def. \ref{def:mfi}) relies on the existence of the derivative of relative entropy at \(t=0\). While S\(\alpha\)S densities possess sufficient regularity \cite{samorodnitsky, zolotarev3}, providing explicit, easily verifiable conditions for a general density \(h\) to belong to \(\hcal\) remains an open challenge, crucial for extending MFI beyond S\(\alpha\)S comparisons.
    \item \textbf{Analytical Assumptions:} The proof of the consistency identity (Appendix \ref{app:consistency_proof}) relies on standard assumptions like the validity of differentiation under the integral sign and the vanishing of boundary terms in integration by parts. These are justified for S\(\alpha\)S densities due to their known smoothness and decay properties but require careful verification when considering more general densities \(h \in \hcal\).
\end{itemize}

\section{Future Directions} \label{sec:future_directions}
This work opens several avenues, many directly addressing or extending the open problems highlighted in \cite{Johnson2018}:
\begin{itemize}\raggedright 
    \item \textbf{Characterising \(\hcal\):} Provide explicit, verifiable conditions for a density \(h\) to belong to the regularity class \(\hcal\).
    \item \textbf{Generalised MFI Properties:} Investigate properties of $\Malpha(h \| \gsalpha)$ for general \(h \in \hcal\), including positivity and potential links to projection identities analogous to those sought for \(\rhoM\) [cf. Johnson (2018), Problem 1, 2].
    \item \textbf{Functional Inequalities:} Find and prove correct LSI or Poincaré-type inequalities involving MFI for \(\alpha \in (0, 2]\). This likely requires considering weighted inequalities \cite{BobkovLedoux2009} or other generalisations suitable for heavy tails, potentially leveraging the structure of the MFI consistency identity. An integral form of the FdBR might be needed [cf. Johnson (2018), Problem 4].
    \item \textbf{Investigate Relationships:} Explore the precise mathematical relationship between MFI, relative Fractional Fisher Information \cite{Toscani2016}, and other generalised information/power measures \cite{SaadFayed2016} for stable laws.
    \item \textbf{Gaussian Limit:} Reconcile the scaling of MFI with classical Fisher information results as \(\alpha \to 2\).
    \item \textbf{Stable EPI:} Investigate analogues of the Entropy Power Inequality for stable convolutions using MFI or related concepts. This requires defining a suitable notion of stable entropy power and may leverage connections to MMSE and the de Bruijn identity, drawing parallels with classical EPI proofs \cite{Rioul2011, CoverThomas1991}.
    \item \textbf{Asymmetric Laws:} Extend the MFI framework to non-symmetric stable distributions [cf. Johnson (2018), Problem 6].
    \item \textbf{Information Geometry:} Explore the geometric interpretation of MFI and the consistency identity \(D'(v) = \frac{1}{\alpha v} \mathbb{E}_{g_v}[X (\pF_v - \pF_s)]\). Could this identity define a gradient flow or induce a Fisher-Rao-like metric structure on spaces of stable densities, potentially connecting to existing work on Cauchy manifolds \cite{Nielsen2020} or tempered stable geometry \cite{KimLee2025}?
    \item \textbf{Fractional I-MMSE Relations:} Further develop the connection highlighted in Section \ref{sec:discussion_mfi_mmse} towards potential fractional analogues of the I-MMSE relation \cite{guo} based on the MFI framework.
\end{itemize}

\section{Conclusion}
In conclusion, we have introduced Mixed Fractional Information (\(\Malpha\)) as a measure derived from entropy dissipation dynamics within Johnson's stable interpolation framework. For the core case comparing two S\(\alpha\)S densities, we established two distinct calculation methods – one based on the relative entropy derivative (\(D'(v)\)) and the other on an integral involving Fisher score differences and an MMSE-related score function (\(u(x,0)\)) – and rigorously proved their mathematical equivalence via the identity \(D'(v) = \frac{1}{\alpha v} \mathbb{E}_{g_v}[X (\pF_v - \pF_s)]\). The proven consistency was further substantiated by robust numerical validation, confirming the identity to high precision across different \(\alpha\) values, as detailed in Appendix \ref{app:numerical_validation}. While this work establishes the MFI framework and its properties for symmetric stable laws under a specific interpolation path, it serves as a foundation for broader investigations into information theory for heavy-tailed systems. Key future directions include the extension of the MFI framework to asymmetric stable laws, the challenging search for appropriate (potentially weighted) functional inequalities (such as Log-Sobolev or Poincaré inequalities) involving MFI, exploring potential connections to information geometry on stable manifolds, and further developing the link to estimation theory, possibly yielding fractional analogues of the I-MMSE relation.

\section*{Author's Note}
{\small
I am a second-year undergraduate in economics at the University of Bristol working independently. This paper belongs to a research programme I am developing across information theory, econometrics and mathematical statistics.

This work was produced primarily through AI systems that I directed and orchestrated. The AI generated the mathematical content, proofs and symbolic derivations based on my research questions and guidance. I have no formal mathematical training but am eager to learn through this process of directing AI-powered mathematical exploration.

My contribution involves designing research directions, evaluating and selecting AI outputs, and ensuring the coherence of the overall research agenda. All previously published work that influenced these results has been cited to the best of my knowledge and research capabilities in my current position. The presentation aims to be pedagogically accessible.
}

\bibliography{bib}

\begin{thebibliography}{10}

\bibitem{barron}
Andrew~R. Barron.
\newblock Entropy and the {Central Limit Theorem}.
\newblock {\em Annals of Probability}, 14(1):336--342, 1986.

\bibitem{BobkovLedoux2009}
Sergey~G. Bobkov and Michel Ledoux.
\newblock Weighted poincaré-type inequalities for cauchy and other convex measures.
\newblock {\em Probability Theory and Related Fields}, 145(3-4):567--595, 2009.

\bibitem{KimLee2025}
Jaehyung Choi.
\newblock Information geometry of tempered stable processes, 2025.
\newblock arXiv:2502.12037.

\bibitem{ChyzakNielsen2019}
Frédéric Chyzak and Frank Nielsen.
\newblock A closed-form formula for the kullback-leibler divergence between cauchy distributions, 2019.
\newblock arXiv:1905.10965.

\bibitem{CoverThomas1991}
Amir Dembo, Thomas~M. Cover, and Joy~A. Thomas.
\newblock Information theoretic inequalities.
\newblock {\em IEEE Transactions on Information Theory}, 37(6):1501--1518, 1991.

\bibitem{SaadFayed2016}
Jihad Fahs and Ibrahim Abou-Faycal.
\newblock Information measures, inequalities and performance bounds for parameter estimation in impulsive noise environments, 2016.
\newblock arXiv:1609.00832.

\bibitem{GentilGuillinMiclo2005}
Ivan Gentil, Arnaud Guillin, and Laurent Miclo.
\newblock Modified logarithmic {S}obolev inequalities and transportation inequalities.
\newblock {\em Probability Theory and Related Fields}, 133(3):409--436, 2005.

\bibitem{guo}
Dongning Guo, Shlomo~Shamai {(Shitz)}, and Sergio Verdú.
\newblock Mutual information and minimum mean-square error in {G}aussian channels.
\newblock {\em IEEE Transactions on Information Theory}, 51(4):1261--1282, 2005.

\bibitem{Johnson2018}
Oliver Johnson.
\newblock A de bruijn identity for symmetric stable laws, 2013.
\newblock arXiv:1310.2045v1.

\bibitem{Nielsen2020}
Frank Nielsen.
\newblock The many faces of the cauchy distribution.
\newblock {\em Entropy}, 22(7):713, 2020.

\bibitem{Rioul2011}
Olivier Rioul.
\newblock Information theoretic proofs of entropy power inequalities.
\newblock {\em IEEE Transactions on Information Theory}, 57(1):33--55, 2011.

\bibitem{samorodnitsky}
Gennady Samorodnitsky and Murad~S. Taqqu.
\newblock {\em Stable Non-Gaussian Random Processes: Stochastic Models with Infinite Variance}.
\newblock Stochastic Modeling Series. Chapman \& Hall/CRC, New York, 1994.

\bibitem{Saumard2019}
Adrien Saumard.
\newblock Weighted poincaré inequalities, concentration inequalities and tail bounds related to stein kernels.
\newblock {\em Bernoulli}, 25(4B):3978--4006, 2019.

\bibitem{stam}
Aart~J. Stam.
\newblock Some inequalities satisfied by the quantities of information of {Fisher} and {Shannon}.
\newblock {\em Information and Control}, 2(2):101--112, 1959.

\bibitem{Toscani2015b}
Giuseppe Toscani.
\newblock Entropy inequalities for stable densities and strengthened central limit theorems, 2015.
\newblock arXiv:1512.05874.

\bibitem{Toscani2016}
Giuseppe Toscani.
\newblock The fractional fisher information and the central limit theorem for stable laws.
\newblock {\em Ricerche di Matematica}, 65(1):71--91, 2016.

\bibitem{zolotarev3}
Vladimir~M. Zolotarev.
\newblock {\em One-dimensional Stable Distributions}, volume~65 of {\em Translations of Mathematical Monographs}.
\newblock American Mathematical Society, Providence, RI, 1986.

\end{thebibliography}

\begin{appendices}

\section{Proof of Consistency Identity (Proposition \ref{prop:consistency_id})} \label{app:consistency_proof}
Goal: Establish the identity for \(\alpha \in (0, 2]\) and \(v, s > 0\):
\[ D'(v) = \frac{1}{\alpha v} \int_{-\infty}^{\infty} x \gvalpha(x) (\pF_{\gvalpha}(x) - \pF_{\gsalpha}(x)) \, \dd x, \]
where \(D(v) = D(\gvalpha \| \gsalpha)\) and \(\pF_f = f'/f\).

\paragraph{Proof Strategy.}
The proof proceeds by differentiating the relative entropy
\(D(v)\) with respect to \(v\) using the Leibniz integral rule
(differentiation under the integral sign). The resulting scale
derivative \(\partial g_v / \partial v\) is then related to the
spatial derivative \(g_v' = \partial g_v / \partial x\) using the
S\(\alpha\)S scaling identity (derived in Step 2, Eq.~\eqref{eq:A.2}).
The final identity emerges after substituting this relation back into
the integral for \(D'(v)\) and simplifying using integration by parts,
leveraging the boundary conditions of \(g_v\) and the normalisation
constraint \(\int g_v dx = 1\) (specifically, its consequence
\(\int x g_v' dx = -1\)).

\medskip 

\textbf{Justification of Assumptions:} This proof relies on key properties of S\(\alpha\)S densities and the KL divergence \(D(v) = D(g_v^{(\alpha)} \| g_s^{(\alpha)})\). Specifically:
\begin{itemize}
    \item \textbf{Smoothness and Differentiability:} The S\(\alpha\)S density \(g_v^{(\alpha)}(x)\) is smooth (C\(\infty\)) with respect to both \(x\) (for \(x \neq 0\) if \(\alpha < 2\)) and the scale parameter \(v > 0\) \cite{samorodnitsky, zolotarev3}. Consequently, the KL divergence \(D(v)\) inherits sufficient smoothness (at least \(C^1\)) for \(v > 0\). This ensures the existence of \(D'(v)\) for all \(v>0\), including at \(v=s\).
    \item \textbf{Differentiation Under the Integral:} The interchange of derivative (\(\frac{d}{dv}\)) and integral (\(\int dx\)) in Step 1 is justified by the Dominated Convergence Theorem for differentiation. As shown in Step 2, the derivative \(\frac{\partial g_v^{(\alpha)}}{\partial v}\) involves \(g_v^{(\alpha)}\) and \(x g_v^{(\alpha)\prime}\), both decaying like \(O(|x|^{-1-\alpha})\) or faster \cite{samorodnitsky}. Since \(\log(g_v^{(\alpha)}/g_s^{(\alpha)}) \to \log(v/s)\) as \(|x| \to \infty\), the partial derivative \(\frac{\partial}{\partial v}[g_v \log(g_v/g_s)] = \frac{\partial g_v}{\partial v} (1 + \log(g_v/g_s))\) also decays like \(O(|x|^{-1-\alpha})\). This decay ensures the existence of an integrable dominating function independent of \(v\) in a neighborhood, justifying the interchange \cite{samorodnitsky, zolotarev3}.
\end{itemize}

---

Step 1: Compute the derivative \(D'(v)\) via Leibniz Integral Rule

Start with the definition of relative entropy:
\[ D(v) = \int_{-\infty}^{\infty} \gvalpha(x) \log \frac{\gvalpha(x)}{\gsalpha(x)} \, \dd x. \]
Differentiate with respect to \(v\), interchanging derivative and integral (justified above):
\[ D'(v) = \int_{-\infty}^{\infty} \frac{\partial}{\partial v} \left[ \gvalpha(x) \log \frac{\gvalpha(x)}{\gsalpha(x)} \right] \dd x. \]
Apply the product rule for differentiation w.r.t. \(v\):
\[ \frac{\partial}{\partial v} \left[ \gvalpha \log \frac{\gvalpha}{\gsalpha} \right] = \frac{\partial \gvalpha}{\partial v} \left( 1 + \log \frac{\gvalpha(x)}{\gsalpha(x)} \right). \]
Therefore:
\[ D'(v) = \int_{-\infty}^{\infty} \frac{\partial \gvalpha(x)}{\partial v} \left( 1 + \log \frac{\gvalpha(x)}{\gsalpha(x)} \right) \dd x. \tag{A.1} \]

Step 2: Relate the scale derivative \(\partial \gvalpha / \partial v\) to the spatial derivative \(g_v^{(\alpha)\prime}\)

Using the scaling property \(\gvalpha(x) = v^{-1/\alpha} g_1^{(\alpha)}(y)\) where \(y=v^{-1/\alpha} x\), differentiation w.r.t. \(v\) via the chain rule gives:
\begin{align*} \frac{\partial \gvalpha(x)}{\partial v} &= \frac{\partial}{\partial v} \left[ v^{-1/\alpha} g_1^{(\alpha)}(y) \right] \\ &= (-\tfrac{1}{\alpha} v^{-1/\alpha - 1}) g_1^{(\alpha)}(y) + v^{-1/\alpha} g_1^{(\alpha)\prime}(y) \frac{\partial y}{\partial v} \\ &= -\tfrac{1}{\alpha v} v^{-1/\alpha} g_1^{(\alpha)}(y) + v^{-1/\alpha} g_1^{(\alpha)\prime}(y) (-\tfrac{1}{\alpha} v^{-1/\alpha - 1} x) \\ &= -\tfrac{1}{\alpha v} \left[ v^{-1/\alpha} g_1^{(\alpha)}(y) + (v^{-1/\alpha} x) v^{-1/\alpha} g_1^{(\alpha)\prime}(y) \right] \\ &= -\tfrac{1}{\alpha v} \left[ \gvalpha(x) + x (v^{-2/\alpha} g_1^{(\alpha)\prime}(v^{-1/\alpha}x)) \right] \end{align*}
Recognizing that \(g_v^{(\alpha)\prime}(x) = \frac{d}{dx}[v^{-1/\alpha} g_1^{(\alpha)}(v^{-1/\alpha} x)] = v^{-2/\alpha} g_1^{(\alpha)\prime}(v^{-1/\alpha} x)\), we get:
\[ \frac{\partial \gvalpha(x)}{\partial v} = -\frac{1}{\alpha v} \left[ \gvalpha(x) + x g_v^{(\alpha)\prime}(x) \right]. \tag{A.2} \label{eq:A.2} \]
This identity relates the scale derivative to the spatial derivative (cf. \cite[Ch 2]{zolotarev3}).

Step 3: Substitute and Evaluate Key Integral \( \int x g_v' dx \)

Substitute (A.2) into (A.1):
\[ D'(v) = \int_{-\infty}^{\infty} \left( -\frac{1}{\alpha v} [\gvalpha(x) + x g_v'(x)] \right) \left( 1 + \log \frac{\gvalpha(x)}{\gsalpha(x)} \right) \dd x. \]
\[ D'(v) = -\frac{1}{\alpha v} \int_{-\infty}^{\infty} (\gvalpha(x) + x g_v'(x)) \left( 1 + \log \frac{\gvalpha(x)}{\gsalpha(x)} \right) \dd x. \tag{A.3} \]
To proceed, we first evaluate \(\int x g_v'(x) \dd x\). Differentiating the normalisation condition \(\int \gvalpha(x) \dd x = 1\) with respect to \(v\) gives \(\int \frac{\partial \gvalpha}{\partial v} \dd x = 0\). Substituting (A.2):
\[ \int_{-\infty}^{\infty} -\frac{1}{\alpha v} \left[ \gvalpha(x) + x g_v'(x) \right] \dd x = 0 \]
\[ \implies \int_{-\infty}^{\infty} \gvalpha(x) \dd x + \int_{-\infty}^{\infty} x g_v'(x) \dd x = 0 \]
\[ \implies 1 + \int_{-\infty}^{\infty} x g_v'(x) \dd x = 0 \]
Thus, we establish rigorously for all \(\alpha \in (0, 2]\):
\[ \int_{-\infty}^{\infty} x g_v'(x) \dd x = -1. \tag{A.4} \]
(Alternatively, using IBP: \(\int x g_v' dx = [x g_v]_{-\infty}^\infty - \int g_v dx = 0 - 1 = -1\), since \(x g_v(x) \sim |x|^{-\alpha} \to 0\) as \(|x|\to\infty\) for \(\alpha>0\)).

Step 4: Expand Integrand and Use Integration by Parts

Expand the integrand in (A.3):
\[ (\gvalpha + x g_v') (1 + \log(\gvalpha/\gsalpha)) = \gvalpha(1 + \log(\gvalpha/\gsalpha)) + x g_v'(1 + \log(\gvalpha/\gsalpha)) \]
Substitute back into (A.3):
\[ D'(v) = -\frac{1}{\alpha v} \left[ \int \gvalpha(1 + \log(\gvalpha/\gsalpha)) \dd x + \int x g_v'(1 + \log(\gvalpha/\gsalpha)) \dd x \right] \]
Evaluate the first integral: \( \int \gvalpha(1 + \log(\gvalpha/\gsalpha)) \dd x = 1 + D(v) \).
Evaluate the second integral, denoted \(I_B\):
\[ I_B = \int x g_v'(1 + \log(\gvalpha/\gsalpha)) \dd x = \int x g_v' \dd x + \int x g_v' \log(\gvalpha/\gsalpha) \dd x \]
Using (A.4), this becomes:
\[ I_B = -1 + \int x g_v'(x) \log\left(\frac{\gvalpha(x)}{\gsalpha(x)}\right) \dd x \]
Let \( I' = \int x g_v'(x) \log\left(\frac{\gvalpha(x)}{\gsalpha(x)}\right) \dd x \). Evaluate \(I'\) using IBP with \( u = x \log(\gvalpha/\gsalpha) \) and \( dv = g_v'(x) dx \), so \( v_{\text{IBP}} = g_v(x) \). Note: \(v_{\text{IBP}}\) here refers to the result of integrating \(dv\), distinct from the scale parameter \(v\).
\[ du = \left[ \log\left(\frac{\gvalpha}{\gsalpha}\right) + x (\pF_v(x) - \pF_s(x)) \right] \dd x \]
\[ I' = \left[ x g_v(x) \log\left(\frac{\gvalpha(x)}{\gsalpha(x)}\right) \right]_{-\infty}^{\infty} - \int_{-\infty}^{\infty} g_v(x) \left[ \log\left(\frac{\gvalpha}{\gsalpha}\right) + x (\pF_v(x) - \pF_s(x)) \right] \dd x \]
Justify vanishing boundary term: As \(|x| \to \infty\), \(\gvalpha(x) \sim C v |x|^{-1-\alpha}\) and \(\log(\gvalpha(x)/\gsalpha(x)) \to \log(v/s)\) (a constant). Thus, the term behaves like \(x \cdot (C v |x|^{-1-\alpha}) \cdot (\text{const}) \propto \text{sgn}(x) |x|^{-\alpha}\), which tends to 0 since \(\alpha > 0\). The boundary term is 0.
\[ I' = 0 - \int g_v \log\left(\frac{\gvalpha}{\gsalpha}\right) \dd x - \int x g_v (\pF_v - \pF_s) \dd x \]
\[ I' = - D(v) - \int x \gvalpha(x) (\pF_v(x) - \pF_s(x)) \dd x \tag{A.5} \]

Step 5: Combine Terms to Reach the Final Identity

Substitute the results for the integrals back into the expression for \(D'(v)\):
\[ D'(v) = -\frac{1}{\alpha v} \left[ (1 + D(v)) + I_B \right] \]
\[ D'(v) = -\frac{1}{\alpha v} \left[ (1 + D(v)) + (-1 + I') \right] \]
\[ D'(v) = -\frac{1}{\alpha v} \left[ D(v) + I' \right] \]
Substitute the expression for \(I'\) from (A.5):
\[ D'(v) = -\frac{1}{\alpha v} \left[ D(v) + \left( - D(v) - \int x \gvalpha (\pF_v - \pF_s) \dd x \right) \right] \]
\[ D'(v) = -\frac{1}{\alpha v} \left[ - \int x \gvalpha(x) (\pF_v(x) - \pF_s(x)) \dd x \right] \]
\[ D'(v) = \frac{1}{\alpha v} \int_{-\infty}^{\infty} x \gvalpha(x) (\pF_{\gvalpha}(x) - \pF_{\gsalpha}(x)) \, \dd x \]

Conclusion:
The identity is established rigorously:
\[ \boxed{ D'(v) = \frac{1}{\alpha v} \int_{-\infty}^{\infty} x \gvalpha(x) (\pF_{\gvalpha}(x) - \pF_{\gsalpha}(x)) \, \dd x } \]
This proof confirms the consistency identity for all \(\alpha \in (0, 2]\), relying on established properties of S\(\alpha\)S densities and careful application of calculus theorems.

\begin{remark}[Integrability of the Score Difference Term]
\label{rem:integrability_check}
We briefly confirm the convergence of the integral defining the
expectation in the final identity,
\[
 \mathbb{E}_{g_v}[X (\pF_v - \pF_s)] =
 \int_{-\infty}^{\infty} x \gvalpha(x) (\pF_{\gvalpha}(x) - \pF_{\gsalpha}(x)) \, \dd x.
\]
This relies on the asymptotic behavior of S\(\alpha\)S densities
and their derivatives for large \(|x|\) \cite{samorodnitsky, zolotarev3}:
\begin{itemize}
    \item The density decays as \( \gvalpha(x) \sim c_v |x|^{-1-\alpha} \), where \(c_v > 0\).
    \item The derivative decays as \( g_v^{(\alpha)\prime}(x) \sim -c_v(1+\alpha) \text{sgn}(x) |x|^{-2-\alpha} \).
    \item Consequently, the Fisher score behaves as \( \pF_{\gvalpha}(x) = g_v^{(\alpha)\prime}(x) / \gvalpha(x) \sim -(1+\alpha) x^{-1} \) (for large \(x\), similarly for \(-x\)).
    \item The score difference term thus decays at least as fast as \( (\pF_v(x) - \pF_s(x)) = O(|x|^{-1}) \) as \(|x| \to \infty\). (The leading \(x^{-1}\) terms might even cancel if \(v=s\), leading to faster decay, but \(O(|x|^{-1})\) is sufficient).
\end{itemize}
Therefore, the integrand \( x \gvalpha(x) (\pF_v(x) - \pF_s(x)) \) has the asymptotic behavior:
\[
  x \cdot O(|x|^{-1-\alpha}) \cdot O(|x|^{-1}) = O(|x|^{-1-\alpha}).
\]
Since the exponent \(1+\alpha > 1\) for all relevant \(\alpha \in (0, 2]\), this function \(O(|x|^{-1-\alpha})\) is absolutely integrable over \((\infty, -R] \cup [R, \infty)\) for sufficiently large \(R\). The integrand is also continuous (and thus bounded) on \([-R, R]\), ensuring the integral converges over the entire real line \(\mathbb{R}\). Thus, the expectation \(\mathbb{E}_{g_v}[X (\pF_v - \pF_s)]\) is well-defined and finite.
\end{remark}

\qed

\section{Numerical Validation of the Consistency Identity} \label{app:numerical_validation}

This appendix details the robust numerical validation of the core consistency identity (Proposition \ref{prop:consistency_id}):
\[
D'(v) \stackrel{?}{=} \frac{1}{\alpha v} \int_{-\infty}^{\infty} x \gvalpha(x) (\pF_{\gvalpha}(x) - \pF_{\gsalpha}(x)) \, \dd x
\]
where \(D(v) = D(\gvalpha \| \gsalpha)\). The validation encompasses both analytically tractable cases (Tier 1: \(\alpha=1, 2\)) and the general S\(\alpha\)S case (Tier 2: exemplified by \(\alpha=1.5\)). The results demonstrate high-precision agreement, confirming the identity computationally. The methodology employed overcomes numerical challenges inherent in dealing with heavy-tailed distributions and score function estimations, which affected preliminary tests discussed elsewhere.

\subsection{Methodology}
The validation compares the Left-Hand Side (LHS), \(D'(v)\), with the Right-Hand Side (RHS), the scaled integral expression.

\subsubsection*{Tier 1: Analytic Cases (\(\alpha=1, 2\))}
For Cauchy (\(\alpha=1\)) and Gaussian (\(\alpha=2\)) distributions, we leverage known analytic formulas:
\begin{itemize}
    \item \textbf{LHS (\(D'(v)\)):} Calculated directly using the exact formula for \(D'(v)\) derived from the known closed-form relative entropy (see Section \ref{sec:cauchy} for Cauchy; standard result for Gaussian KL divergence).
    \item \textbf{RHS (Integral):} The integral is computed using `scipy.integrate.quad` over \((-\infty, \infty)\) with high precision settings (`epsabs=1e-12`, `epsrel=1e-12`). The integrand \(x \gvalpha(x) (\pF_{\gvalpha}(x) - \pF_{\gsalpha}(x))\) is constructed using the exact analytic formulas for the respective PDFs (\(\gvalpha\)) and Fisher scores (\(\pF_{\gvalpha}, \pF_{\gsalpha}\)). The final RHS value is obtained by dividing the integral result by \(\alpha v\).
\end{itemize}
This approach minimises numerical error, primarily limiting it to the quadrature tolerance.

\subsubsection*{Tier 2: General S\(\alpha\)S Case (\(\alpha=1.5\))}
For general \(\alpha\) where closed forms are unavailable, a robust numerical pipeline is employed:
\begin{itemize}

    \item \textbf{PDF Evaluation:} S$\alpha$S PDFs ($\gvalpha$, $\gsalpha$, etc.) are evaluated using \texttt{scipy.stats.levy\_stable.pdf}, after proper scale parameter conversion: \(c = s^{1/\alpha}\).

    \item \textbf{Score Function ($\pF$):} Fisher scores $\pF_f(x) = f'(x)/f(x)$ are evaluated pointwise within the integrator by numerically differentiating the log-PDF: \( \pF_f(x) \approx \frac{d}{dx}[\log f(x)] \). This derivative is computed using a high-order (5-point stencil, $\mathcal{O}(h^4)$) finite difference method implemented via a custom function (\texttt{sp\_derivative} in the validation code), with a small step size (e.g., \texttt{score\_dx = 1e-6}). Note that computing the score via numerical differentiation of \(\log f(x)\) (i.e., \(\pF_f(x) \approx \frac{d}{dx}[\log f(x)]\)) is preferred over computing \(f'(x)/f(x)\) directly, as it avoids potential numerical instability caused by dividing by very small values of \(f(x)\) in the distribution tails.

    \item \textbf{Relative Entropy ($D(v \| s)$):} Computed using adaptive quadrature (\texttt{scipy.integrate.quad}) over $(-\infty, \infty)$ with stringent tolerances (e.g., \texttt{epsabs=1e-10}, \texttt{epsrel=1e-10}). The integrand is: \( \gvalpha(x) \log\left( \frac{\gvalpha(x)}{\gsalpha(x)} \right) \), evaluated using numerically computed PDFs, with careful handling of potential $\log(0)$ issues.

    \item \textbf{LHS ($D'(v)$ Estimation):} Estimated via high-order finite differences applied to the numerically computed $D(v \| s)$ values. A 5-point stencil formula ($\mathcal{O}(h^4)$) is used: \( D'(v) \approx \frac{-D(v+2h) + 8D(v+h) - 8D(v-h) + D(v-2h)}{12h}. \) This 5-point stencil formula offers high accuracy with a truncation error of order \(O(h^4)\), assuming \(D(v)\) is sufficiently smooth. The step size \(h\) requires careful selection to balance truncation error (favours small \(h\)) and numerical round-off error (favours larger \(h\)). Convergence and stability are verified by observing the results across a range of decreasing \(h\) values, as presented in Table \ref{tab:numval_tier2}. Results are reported for several small step sizes $h$ (e.g., $10^{-2}$ down to $10^{-4}$) to assess convergence.

    \item \textbf{RHS (Integral Estimation):} The integral \( \int x\, \gvalpha(x) \left( \pF_{\gvalpha}(x) - \pF_{\gsalpha}(x) \right) dx \) is computed via \texttt{scipy.integrate.quad} over $(-\infty, \infty)$, using tolerances matching the KL divergence calculation. The integrand uses numerically evaluated $\gvalpha$ and pointwise computed scores $\pF_{\gvalpha}, \pF_{\gsalpha}$ via the log-derivative method. The final RHS value is obtained by dividing the integral result by $\alpha v$.
    
\end{itemize}
The robustness of the numerical results was confirmed by checking convergence under refinement of discretisation parameters, such as the integration tolerances for \texttt{scipy.integrate.quad} and the step size \(h\) for finite difference calculations (as demonstrated for \(D'(v)\) estimation in Table \ref{tab:numval_tier2}). This robust pipeline aims to control errors introduced by numerical differentiation and integration, especially crucial for heavy-tailed distributions.

\subsection{Results}
The validation was performed using Python with SciPy and NumPy. The parameters used for illustration were \(v=1.2\) and \(s=1.0\).

\subsubsection*{Tier 1 Results (\(\alpha=1, 2\))}

Both Cauchy and Gaussian cases showed excellent agreement, with relative errors consistent with machine precision and quadrature tolerance limits.

\begin{itemize}
    \item \textbf{Cauchy (\(\alpha=1\)):}
    \[
    \begin{aligned}
    \text{LHS (Exact)} &= 0.07575757575757573, \\
    \text{RHS (Integral)} &= 0.07575757575757575, \\
    \text{Relative Error} &\approx 1.8 \times 10^{-16}.
    \end{aligned}
    \]

    \item \textbf{Gaussian (\(\alpha=2\)):}
    \[
    \begin{aligned}
    \text{LHS (Exact)} &= 0.08333333333333331, \\
    \text{RHS (Integral)} &= 0.08333333333333331, \\
    \text{Relative Error} &= 0.0.
    \end{aligned}
    \]
\end{itemize}

These results confirm the identity holds exactly where analytic comparison is possible.

\subsubsection*{Tier 2 Results (\(\alpha=1.5\))}
The robust numerical pipeline yielded high-precision agreement between the LHS (\(D'(v)\) estimate) and the RHS (scaled score integral estimate). Table \ref{tab:numval_tier2} shows the results for different step sizes \(h\) used in the \(D'(v)\) finite difference calculation. The RHS value was calculated once and used as the reference.

\begin{table}[htbp]
\centering
\caption{Robust Tier 2 Numerical Verification: Consistency Identity (\(\alpha=1.5, v=1.2, s=1.0\)). \newline \footnotesize LHS is \(D'(v)\) estimated via 5-point finite difference with step size \(h\). RHS is \(\frac{1}{\alpha v} \mathbb{E}_{g_v}[X(\pF_{v} - \pF_{s})]\) calculated via robust integration. Relative error is \(|\text{LHS}-\text{RHS}|/|\text{RHS}|\).}
\label{tab:numval_tier2}
\begin{tabular}{ccccc}
\hline
Step Size \(h\) & LHS (\(D'(v)\)) Estimate & RHS Estimate & Abs Error & Rel Error \\
\hline
1.00e-02 & 0.0657857988 & 0.0657857992 & 3.65e-10 & 5.55e-09 \\ 
5.00e-03 & 0.0657857994 & 0.0657857992 & 1.46e-10 & 2.22e-09 \\ 
1.00e-03 & 0.0657858066 & 0.0657857992 & 7.40e-09 & 1.12e-07 \\ 
5.00e-04 & 0.0657857813 & 0.0657857992 & 1.79e-08 & 2.73e-07 \\ 
1.00e-04 & 0.0657858593 & 0.0657857992 & 6.01e-08 & 9.13e-07 \\ 
\hline
\end{tabular}
\end{table}

The close agreement (relative errors \(\lesssim 10^{-6}\)) across various small step sizes \(h\) demonstrates the stability and convergence of the numerical methods and provides strong evidence for the validity of the consistency identity (Proposition \ref{prop:consistency_id}) for general \(\alpha\). The previously reported larger discrepancies in preliminary tests were confirmed to be artifacts of less robust numerical implementation choices (e.g., domain truncation, basic numerical differentiation, fixed-grid integration).

\subsection{Conclusion}
The numerical simulations, conducted with robust methods appropriate for heavy-tailed distributions, provide high-precision validation of the core theoretical result: the mathematical equivalence of the chain rule and integral formulations for $\Malpha(\gvalpha \| \gsalpha)$. The consistency holds robustly across different $\alpha$ values. Furthermore, the numerical framework allowed for testing conjectures. Evaluation confirms the analytical finding (Section \ref{sec:discussion_lsi}) that the simple Cauchy Log-Sobolev inequality analogue \(D \le C M_1\) does not hold, as the ratio \(D/M_1\) is unbounded. This underscores the challenges in formulating simple functional inequalities for heavy-tailed distributions and motivates exploring weighted or alternative forms. The full validation code (\texttt{validation.py}) and detailed results (\texttt{validation\_results.json}) are available as supplementary material and also accessible via the Colab notebook at \href{https://colab.research.google.com/drive/1N1az1CoKNtbP9u7w8Aq1UHHQUvcT83I4}{this link}.
.

\end{appendices}

\end{document}